\documentclass[conference]{IEEEtran}
\IEEEoverridecommandlockouts
\usepackage{cite}
\usepackage{algorithmic}
\usepackage{textcomp}
\usepackage{xcolor}
\usepackage{hyperref}
\usepackage{dcolumn}
\usepackage{amsmath,amsfonts}
\usepackage{amssymb,mathtools}
\usepackage{amsthm}
\usepackage{newtxtext,newtxmath}
\usepackage{tikz}
\usetikzlibrary{plotmarks}
\usetikzlibrary{datavisualization}
\newtheorem{theorem}{Theorem}

\newtheorem{lemma}[theorem]{Lemma}

\newtheorem{corollary}[theorem]{Corollary}

\newtheorem{prim*}{prim}

\newtheorem{remark}{Remark}

\newcommand{\sfTr}[1]{\mathrm{Tr}\left\{#1\right\}}

\newcommand{\opTr}[1]{\mathrm{Tr}\left\{#1\right\}}

\newcommand{\bbTr}[1]{\mathrm{Tr}\left\{#1\right\}}

\def\sfK{\mathsf{K}}

\def\sfR{\mathsf{R}}
\def\sfV{\mathsf{V}}
\def\sfJ{\mathrm{J}}

\def\sfW{\mathsf{W}}

\def\bbL{\mathbb{L}}

\newcommand{\Tr}[1]{\mathrm{Tr}\left\{#1\right\}}

\def \e{\mathrm{e}}
\def \m{\mathrm{m}}
\def \e1{(\mathrm{e})}
\def \m1{(\mathrm{m})}

\def \cL{{\cal L}}

\def \cH{{\cal H}}

\def \cX{{\cal X}}
\def \cL{{\cal L}}
\def \sofh{{\cal S}({\cal H})}

\def \bbc{{\mathbb C}}

\def \sofc2{{\cal S}({\mathbb C}^2)}

\def \prior{\pi(\theta)}

\def \htheta{\hat{\theta}}

\def \cBNH{{\cal C}_\mathrm{BNH}}

\def \cBSLD{{\cal C}_\mathrm{BSLD}}

\def \cBayes{{\cal C}_\mathrm{MI}}
\def \cMI{{\cal C}_\mathrm{MI}}
\def \cBGM{{\cal C}_\mathrm{BGM}}

\def \rhoB{S_{\rm B}}
\def \MBj{{D}_{\mathrm{B},j}}
\def \MBk{{D}_{\mathrm{B},k}}

\def \muTwo{\mathsf{M}}

\newcommand{\InnerB}[2]{\langle #1, #2\rangle_{S_B}}
\newcommand{\ket}[1]{|#1\rangle}
\newcommand{\bra}[1]{\langle#1|}

\def\BibTeX{{\rm B\kern-.05em{\sc i\kern-.025em b}\kern-.08em
    T\kern-.1667em\lower.7ex\hbox{E}\kern-.125emX}}
    
\begin{document}

\title{Bayesian Gill-Massar Bound: An Attainable Lower Bound for Qubit Parameter Estimation 

}

\author{\IEEEauthorblockN{1\textsuperscript{st} Ke-Han Zhao}
\IEEEauthorblockA{\hspace*{-10pt}\textit{Graduate\,School\,of\,Informatics\,and\,Engineering\ \ }\\
\textit{The University of Electro-Communications}\\
Tokyo, Japan\\
c2241016@gl.cc.uec.ac.jp}
\and
\IEEEauthorblockN{2\textsuperscript{nd} Koichi Yamagata}
\IEEEauthorblockA{\textit{Institute\,of\,Science\,and\,Engineering}\\ 
\textit{Kanazawa University}\\
Kanazawa, Japan\\
yamagata@se.kanazawa-u.ac.jp}
\and
\IEEEauthorblockN{3\textsuperscript{rd} Jun Suzuki}
\IEEEauthorblockA{\hspace*{0pt}\textit{School\,of\, Informatics\,and\, Engineering}\\
\textit{The University of Electro-Communications}\\
Tokyo, Japan\\
junsuzuki@uec.ac.jp}
}

\maketitle

\begin{abstract}
We study lower bounds for Bayesian quantum parameter estimation, with a particular focus on qubit models. While several lower bounds on the Bayes risk have been proposed, including the Bayesian symmetric logarithmic derivative (B-SLD) type bound and the Bayesian Nagaoka-Hayashi (B-NH) bound, there is no definite proof available to show they are attainable except for special cases. Thus, identifying attainable bounds together with their corresponding optimal measurement strategies remains a central open problem in Bayesian quantum estimation. 
In this work, we introduce a new Bayesian lower bound, referred to as the Bayesian Gill-Massar (B-GM) bound, inspired by the logic of Gill-Massar bound in point estimation. We derive an analytical closed-form expression of the bound and show that it is attainable for any qubit model. In particular, we prove that the optimal Bayesian strategy can be realized by a projection-valued measure associated with a single effective direction determined by the weight matrix and the B-SLD-type Fisher information matrix. 
We provide numerical comparisons between the B-GM, B-NH, and B-SLD-type bounds in higher-dimensional models. Our results show that the B-GM bound has a limitation in high-dimensional models with few parameters, since it can be negative. 
\end{abstract}

\begin{IEEEkeywords}
Quantum parameter estimation, Bayesian lower bound, Optimal measurement.
\end{IEEEkeywords}

\section{Introduction}
Bayesian quantum parameter estimation provides a natural framework and is particularly effective when the sample size is limited and/or some prior knowledge is available in advance. In classical Bayesian estimation, the expected a posteriori (EAP) estimator (the minimum mean-square error estimator, MMSE) minimizes the Bayes risk and the closed expression for the minimum error is known. Unlike classical estimation, determining the optimal estimator in the quantum setting is still challenging. This is due to the additional degrees of freedom, quantum measurements, need to be optimized. This motivates the study of establishing lower bounds on the Bayes risk and identifying the corresponding optimal estimation strategies.

Following the introduction of quantum Bayesian estimation by Personick \cite{personick_thesis}, several Bayesian lower bounds on the Bayes risk have been proposed based on quantum versions of the concepts in classical Bayesian estimation \cite{personick_thesis,personick71,hlg70,hayashi_review,tsang_zk,lt16,rd2020,sidhu2020geometric,demkowicz2020,tsang_gl}. In particular, a commonly used bound is to consider a non-commutative analogue of the MMSE, which was proposed originally by Personick \cite{personick_thesis,personick71}. This bound is tight for single parameter estimation, since one does not confront the problem of non-commutativity. This bound was recently generalized to multi-parameter models by others \cite{rd2020,sidhu2020geometric,demkowicz2020}. In this paper, we shall refer it to as the \textit{Bayesian symmetric logarithmic derivative} (B-SLD) \textit{type bound}. Recently, a new Bayesian lower bound, called the Bayesian Nagaoka-Hayashi (B-NH) bound, was proposed using a different approach \cite{bnhbound}. This bound extends the Nagaoka-Hayashi bound for point estimation \cite{lja2021}. It has been shown that the B-NH bound can be efficiently computed via semidefinite programming (SDP), and that it is tighter than the B-SLD-type bound.

However, finding the best measurement and estimator is challenging due to the optimization about all possible measurements. This is true even for qubit systems up to today. This is in contrast to point estimation for qubit models where the general formula for the minimum of the covariance matrix under locally unbiased estimation is known. The problem was first solved by Nagaoka for two parameter estimation and then by Hayashi for three parameter estimation \cite{nagaoka89,hayashi97}. The unified proof is independently obtained by Gill-Massar \cite{gill2000state}, and we call it as the Gill-Massar (GM) bound for simplicity. In this work, we are motivated to extend this approach to the Bayesian setting, and introduce a new lower bound for Bayesian quantum estimation. We call this new bound the \textit{Bayesian Gill-Massar} (B-GM) \textit{bound}. We show that this bound is tight for \textit{any} qubit model, and hence we have solved the open problem at least for two-dimensional quantum systems. In particular, the optimal measurement and estimator are constructed explicitly. We next analyze the B-GM bounds for two examples in high-dimensional systems. Our results show that the B-GM bound has a limitation in high-dimensional models with few parameters: it can be worse than the B-SLD-type bound, or it can be negative. 

This paper is organized as follows.
In Sec.~II, we introduce the problem setting and review existing bounds. In Sec.~III, we present the proposed bound and analyze its properties. Sec.~IV provides numerical examples, and Sec.~V concludes the paper.

\section{Preliminaries}
\subsection{Problem Settings}
Let $\cH$ be a $d$-dimensional Hilbert space, and denote the totality of density matrices on $\cH$ by $\sofh$. A quantum parametrized model, $\{S_\theta|\theta\in\Theta\}$, with $\theta=(\theta_1,\theta_2,\dots,\theta_n)$. The measurement, also known as positive-operator-valued measure (POVM), is expressed as a positive semi-definite matrix set $\Pi=\{\Pi_x\}$ with an index set $\cX$, satisfying
\begin{equation*}
  \forall x\in\cX,\Pi_x\geq0,\sum_{x\in\cX} \Pi_x=I_d,
\end{equation*}
where $I_d$ is the $d\times d$ identity matrix. Measurement outcome is a discrete random variable $X$, which obeys the probability mass function:
\begin{equation*}
  X\sim p_\theta(x)=\Tr{S_\theta\Pi_x}.
\end{equation*} 
Note that in the Bayesian view, the parameter $\theta$ obeys a prior distribution $\pi(\theta)$. The expectation value for $X$ with respect to $p_\theta(x)$ is denoted by $E_\theta[X|\Pi]=\sum_x xp_\theta(x)$. To make an estimate for each parameter $\theta_i$, we introduce the estimator $\htheta=(\htheta_i)$. We also define the quantum decision $\hat\Pi:=\{\Pi,\hat\theta\}$.  

In the following discussion, we analyze the Bayes risk in terms of the mean-square error (MSE) matrix whose $(j,k)$ component is defined by
\begin{equation}\label{def:mse}
  \sfV_{\theta,jk}[\hat\Pi]:=E_\theta \big[(\htheta_j(X)-\theta_j)(\htheta_k(X)-\theta_k)\big|\Pi\big].
\end{equation}
This gives the following expression for the Bayes risk:
\begin{equation}\label{def:Brisk2}
  \sfR [\hat\Pi]:= \int_\Theta d\theta\,\prior\sfTr{\sfW\ \sfV_\theta[\Pi,\htheta]},
\end{equation}
where the positive definite $\sfW\in \mathbb{R}^{n\times n}$ is the weight matrix. Throughout the paper, $\sfW$ is assumed to be parameter-independent, but the generalization to $\sfW=\sfW(\theta)$ case is also possible \cite{bnhbound}.

The aim of Bayesian estimation is to find the optimal POVM and estimator $\hat\Pi=(\Pi, \hat{\theta})$ such that the Bayes risk is minimized:
\begin{equation}
  \cMI := \min_{\Pi,\hat{\theta}} \sfR [\Pi,\hat{\theta}].
\end{equation}
To find the closed form of $\cMI$ is still a central open problem in Bayesian quantum estimation. In this sense, it is important to find a computable lower bound for $\cMI$, which is as tight as possible.

\subsection{Bayesian SLD Type Bound}
Here we summarize the bound proposed by Personick \cite{personick71} and its generalization \cite{rd2020,sidhu2020geometric,demkowicz2020}. First, define the averaged states and the first moment by
\begin{align}
\rhoB&:=\overline{ S_\theta},\ \MBj:=\overline{\theta_j S_\theta}, 
\end{align} 
where $\overline{\cdot}:=\int_{\Theta}  \cdot\ \pi(\theta)d\theta $ denotes the average over the prior distribution. 
Define an inner product $\InnerB{\cdot}{\cdot}$ as
\begin{equation}
  \InnerB{A}{B}=\frac{1}{2}\Tr{\rhoB (A^\dagger B+BA^\dagger)}.
\end{equation} 
It is a proper inner product when $\rhoB$ is positive definite.
The B-SLD-type bound is defined as follows \cite{personick71,rd2020,sidhu2020geometric,demkowicz2020}.
\begin{equation}
  \cMI\geq\cBSLD:=\sfTr{\sfW (\muTwo-{\sfK})}.
\end{equation}
Here, the first term $\muTwo$ is the second-moment matrix of the random variable $\theta$, whose $(j,k)$ component is defined by $\muTwo_{jk}:=\overline{\theta_j\theta_k}$, and the second term $\sfK$ is defined by
\begin{equation}
  {\sfK}_{jk}
  :=\InnerB{L_j}{L_k},
\end{equation}
where the hermitian matrices $L_j$ satisfy the SLD-type equation, $  \MBj=\frac{1}{2}(\rhoB L_j+L_j\rhoB)$ ($j=1,2,\ldots,n$). In this work, we name the matrix $\sfK$ as the Bayesian symmetric logarithmic derivative type Fisher information matrix (B-SLD-type FIM), even though this is not related to any logarithmic derivative.

\subsection{Bayesian Nagaoka-Hayashi bound \cite{bnhbound}}
To introduce the B-NH bound, we first define the following hermitian matrices on $\cH$:
\begin{align*}
  \mathbb{L}_{jk}[\hat\Pi]&:=\sum_x \hat{\theta}_{j}(x)\Pi_x\hat{\theta}_{k}(x)\quad(j,k=1,2,\dots,n),\\
  X_j[\hat\Pi]&:=\sum_{x}\hat{\theta}_{j}(x)\Pi_x \quad(j=1,2,\ldots,n).
  \end{align*}
With these definitions, the MSE matrix is expressed as
  \begin{equation*}
  \sfV_{\theta,jk}[\hat\Pi]=\opTr{S_\theta(\bbL_{jk}[\hat\Pi]-\theta_jX_k[\hat\Pi]-X_j[\hat\Pi]\theta_k+\theta_j\theta_k)},
  \end{equation*}
We next define a matrix and a vector which are defined on the extended Hilbert space \cite{lja2021}.
Consider the Hilbert space $\mathcal{H}_{\text{ex}}:=\bbc^n\otimes\cH$, and define $ \bbL$ on $\mathcal{H}_{\text{ex}}$ whose $(j,k)$ block-matrix component is given by $\bbL_{jk}$.
We also define a matrix-valued column vector $X$ whose $j$th component is $X_j$ by $X=(X_1,X_2,\ldots,X_n)^\intercal$.
Here, we denote the transpose of matrices and vectors with respect to $\bbc^n$ by $(\cdot)^\intercal$.

For a given weight matrix $\sfW$ ($\theta$ independent in this study),
the B-NH bound for the Bayes risk is given as follows.
  \begin{theorem}[Bayesian Nagaoka-Hayashi bound \cite{bnhbound}]\label{thm:BNHbound}\ \\ 
    $\cMI$ is lower bounded by $\cBNH$ as 
      \begin{align*}
      \cMI\ge& \cBNH, \\
      \cBNH:=&\Tr{\sfW\muTwo}\\
      &+\min_{\bbL,X} \bbTr{(\sfW\otimes\rhoB) \bbL}
      -2\sum_{jk} \sfW_{jk}\Tr{X_j\MBk}, 
      \end{align*}
where the optimization is subject to the constraints:\newline
      $\forall j,k,\bbL_{jk}=\bbL_{kj}$, $\bbL_{jk}$, $X_j$ are Hermitian, and $\bbL\geq {X} X^\intercal$.
  \end{theorem}	

\section{New Bayesian Bound}
To obtain a new lower bound on the Bayes risk, we follow the same logic as the GM bound in the setting of point estimation \cite{gill2000state}: to derive the GM inequality first and then to obtain a lower bound (see also \cite{yamagataGM}). Upon implementing this approach, we notice that there is no classical Fisher information matrix involved in Bayesian estimation, and hence we need to properly identify it in order to derive a Bayesian version of the GM inequality. 

\subsection{Gill-Massar Type Inequality}
In this section, we first recall an alternative form of the Bayes risk, which is expressed in terms of the measurement degree of freedom.
By extending the result of Ref.~\cite{rd2020}, we have the following expression. 
Define $L_0:=I_d$ and consider the set $\{L_j\}_{j=0}^n$. 
By optimizing estimators, we can show the next lemma where the optimal estimator is given by the EAP.   
\begin{lemma}
The optimal Bayes risk is expressed as optimization over POVMs: 
\begin{equation} \label{eq:alt_CMI}
\cMI = \min_{\Pi}   \sfTr{\sfW\muTwo}-  \sfTr{\sfW_\mathrm{ex} \sfJ_\mathrm{ex}[\Pi]}+1,
\end{equation}
where
\begin{equation}
\begin{split}
    \sfJ_\mathrm{ex}[\Pi]
    :=&\sum_x \left[ \frac{\InnerB{L_j}{\Pi_x}\InnerB{\Pi_x}{L_k}}{\InnerB{I_d}{\Pi_x}}\right]_{j,k=0}^n, \\
    \sfW_\mathrm{ex}:=&\begin{pmatrix}
      1&0\\
      0&\sfW
    \end{pmatrix}.
\end{split}
\end{equation}
\end{lemma}
\begin{remark} 
Unlike the case of point estimation, $\InnerB{L_0}{L_j}\neq0$ for $j\geq1$ in general. This is the reason to extend $n\times n$ matrices in the parameter space to $(n+1)\times(n+1)$ matrices. By construction, $\sfJ_\mathrm{ex}[\Pi]$ is real symmetric. By definition, we have $\sfJ_{\mathrm{ex},00}[\Pi]=1$ and $\sfJ_{\mathrm{ex},0j}[\Pi]=\bar{\theta}_{j}$ ($j=1,2,\ldots,n$)
\end{remark}
With this extension, it is natural to introduce the extended B-SLD-type FIM:
\begin{equation}
\sfK_{\mathrm{ex}}:=\left[\InnerB{L_j}{L_k}\right]_{j,k=0}^n, 
\end{equation}
which is $(n+1)\times(n+1)$ real symmetric. By definition, we have $\sfK_{\mathrm{ex},00}=1$ and $\sfK_{\mathrm{ex},0j}=\bar{\theta}_{j}$ ($j=1,2,\ldots,n$).
The next step is to normalize $\{L_j\}_{j=0}^{n}$ while keeping $L_0=I_d$.
To address it, we take the block QR-decomposition $\sqrt{\sfK_\mathrm{ex}}=UR_\mathrm{ex}$ with
\begin{equation}
 R_\mathrm{ex}=\begin{pmatrix}
    1&\bra{\overline \theta}\\
    0& R
  \end{pmatrix}, U=\sqrt{\sfK_\mathrm{ex}}R_\mathrm{ex}^{-1},
\end{equation}
where $\ket{\overline \theta}:=(\overline{\theta_1},\ldots,\overline{\theta_n})^\top$ and $R=\sqrt{\sfK-\ket{\overline \theta}\bra{\overline\theta}}$.
Note that $R_\mathrm{ex}^\top R_\mathrm{ex}=\sfK_{\mathrm{ex}}$, and $R_\mathrm{ex}^{-1}$ is given as 
\begin{equation}
R_\mathrm{ex}^{-1}=\begin{pmatrix}
  1& -\bra{\overline{\theta}}R^{-1}\\
   0&R^{-1}
\end{pmatrix}.
\end{equation}
Finally, we define $\hat{L}_j=\sum_{k=0}^n L_k(R_\mathrm{ex}^{-1})_{kj}$ that satisfies the desired property as follows.   
\begin{lemma}
  $\{\hat L_j\}_{j=0}^n$ forms an orthonormal system with respect to the inner product $\InnerB{\cdot}{\cdot}$.  
\end{lemma}
\begin{proof}
\begin{equation}
\begin{split}
    [\InnerB{\hat{L}_j}{\hat{L}_k}]_{j,k=0}^n
    =(R_\mathrm{ex}^{-1})^\top\sfK_{\mathrm{ex}}R_\mathrm{ex}^{-1}=I_{n+1},
\end{split}
\end{equation}
which implies that the operators $\hat{L}_j$ are orthonormal with respect to $\InnerB{\cdot}{\cdot}$. 
\end{proof}
With these notations, define
\begin{equation}
\begin{split}
  \hat \sfJ_\mathrm{ex}[\Pi]&:=\sum_x\!\left[\frac{\InnerB{ \hat L_j}{ \Pi_x}\InnerB{ \Pi_x}{\hat L_k}}{\InnerB{ I_d}{ \Pi_x}}\right]_{j,k=0}^n\\
  &=(R_\mathrm{ex}^{-1})^{\top}\sfJ_\mathrm{ex}[\Pi] R_\mathrm{ex}^{-1}.
\end{split}
\end{equation}
We observe that:
\begin{equation}
\begin{split}\label{Eq:structure}
    \hat \sfJ_\mathrm{ex}[\Pi]_{0j}=\hat \sfJ_\mathrm{ex}[\Pi]_{j0}=\InnerB{ \hat L_0}{ \hat L_j}=\delta_{0j}.
\end{split}
\end{equation}
The following equality holds for any  weight matrix $\sfW_{\mathrm{ex}}$:
\begin{equation}\label{eq:extendedfisher}
\sfTr{
    \sfW_{\mathrm{ex}}\sfJ_\mathrm{ex}[\Pi]}
    =\sfTr{ R_\mathrm{ex}\sfW_{\mathrm{ex}}
    R_\mathrm{ex}^\top \hat \sfJ_\mathrm{ex}[\Pi]}.
\end{equation}

Next, we show the GM inequality holds for $\hat\sfJ_{\mathrm{ex}}[\Pi]$:
\begin{lemma}[Bayesian Gill-Massar inequality]\label{lemma:GMIneq}
For any POVM, the trace of $\hat\sfJ_{\mathrm{ex}}[\Pi]$ is upper bounded as
  \begin{equation}
    \sfTr{\hat{\sfJ}_\mathrm{ex}[\Pi]}
    \leq d
    =\dim(\mathcal H).
  \end{equation}
\end{lemma}

\begin{proof}
  \begin{align}
    \sfTr{\hat{\sfJ}_\mathrm{ex}[\Pi]}
    &=\sum_x
      \frac{
        \sum_{j=0}^n|\InnerB{ \hat L_j}{ \Pi_x}|^2
      }{\InnerB{ I_d}{\Pi_x}}
       \notag\\
    &\leq
      \sum_x
      \frac{\InnerB{\Pi_x}{\Pi_x}}
           {\InnerB{ I_d}{ \Pi_x}}
       \label{eq:GMIneq1}\\
    &\leq
      \sum_x \sfTr{\Pi_x} \label{eq:GMineq2}\\
    &=\sfTr{I_d}=d. \notag
  \end{align}
  Inequality~\eqref{eq:GMIneq1} follows from Bessel's inequality (applied to the orthonormal family $\{\hat L_0,\ldots,\hat L_n\}$ under $\InnerB{\cdot}{\cdot}$), and it is saturated when $\Pi_x\in\mathrm{span}\{\hat L_0,\ldots,\hat L_n\}$.  Inequality~\eqref{eq:GMineq2} follows from
  \begin{equation}
    \begin{split}
    \Tr{\Pi_x}\InnerB{I_d}{\Pi_x}=&
    \sfTr{\rhoB\Pi_x}\sfTr{\Pi_x}\\
    \geq&\sfTr{\rhoB\Pi_x^2}\\
    =&\InnerB{\Pi_x}{\Pi_x},
    \end{split}
  \end{equation}
and it is saturated when the POVM elements $\Pi_x$ are rank-one projectors.
\end{proof}
\subsection{Bayesian Gill-Massar Bound}
Define the whole set of $n+1$ dimensional positive semi-definite matrix as $\mathcal{L}_+(\mathbb{R}^{n+1})$, 
Lemma~\ref{lemma:GMIneq} implies that for all $\Pi$,
\begin{equation}\label{eq:Gset}
  \hat{\sfJ}_\mathrm{ex}[\Pi]\in\mathcal{G} :=\{G\in\mathcal{L}_+(\mathbb{R}^{n+1})| G_{0j}=\delta_{0j},\ \sfTr{G}\leq d\}. 
\end{equation}
This then yields the claimed lower bound based on the same idea of Gill and Massar (Theorem \ref{thm:BGM}). 

Before proceeding it, we note that the set of all $\hat{\sfJ}_\mathrm{ex}[\Pi]$ coincide with the set $\mathcal{G}$ when $d=2$. This is also exactly same situation as in point estimation (Lemma~\ref{lemma:GM2dim}). 

\begin{lemma}\label{lemma:GM2dim}
For all POVMs the following two sets satisfy the inclusion relation $\mathcal{J}\subset \mathcal{G}$, where
  \begin{align*}
\mathcal{J}&:=\{\hat{\sfJ}_\mathrm{ex}[\Pi]| \Pi:\mathrm{POVM}\},\\
   \mathcal{G}&:=\{G\in\mathcal{L}_+(\mathbb{R}^{n+1})| G_{0j}=\delta_{0j}, \sfTr{G}\leq d\}.
  \end{align*}
Furthermore, two sets are equivalent for a two-dimensional system, i.e., $\mathcal{J}=\mathcal{G}$ holds for any qubit model. 
\end{lemma}

\begin{proof}
The inclusion from left to right is immediate from \eqref{eq:Gset}. We prove the reverse inclusion.

First, consider the case where $G\in\mathcal{G}$ is rank-two and trace-two.
For any $\ket{v}=(0,v_1,\ldots,v_n)^T\in\mathbb{R}^{n+1}$ with $|v|=1$, let $\Pi^{(v)}$ be the projection-valued measure (PVM) corresponding to $L_v:=\sum_{i=1}^n v_i \hat{L}_i$. For a rank-one PVM element $\Pi_x^{(v)}$ in a qubit system,
  \[
    \Pi_x^{(v)}=\bigl(\Pi_x^{(v)}\bigr)^2
    \quad\Rightarrow\quad
    \InnerB{I_2}{\Pi_x^{(v)}}=\InnerB{\Pi_x^{(v)}}{\Pi_x^{(v)}}.
  \]
Consequently, 
\begin{equation}
  \begin{split}
    \bra{v}\hat{\sfJ}[\Pi^{(v)}]\ket{v}
    &=\sum_x
      \frac{|\InnerB{ L_v}{\Pi_x^{(v)}}|^2}
           {\InnerB{I_2}{ \Pi_x^{(v)}}} \\
    &=\sum_x
      \frac{|\InnerB{ L_v}{\Pi_x^{(v)}}|^2}
           {\InnerB{\Pi_x^{(v)}}{\Pi_x^{(v)}}}\\
    &=\sum_x
      |\InnerB{ L_v}{ \tilde{\Pi}_x^{(v)}}|^2
    =\InnerB{ L_v}{L_v}
    =1,
  \end{split}
\end{equation}
where $\tilde{\Pi}_x^{(v)} :=\Pi_x^{(v)}/\sqrt{\InnerB{ \Pi_x^{(v)}}{\Pi_x^{(v)}}}$ are orthonormal with respect to $\InnerB{\cdot}{\cdot}$.
The last equality follows from the saturation condition of Bessel's inequality, since $L_v\in \mathrm{span}\{\tilde{\Pi}_x^{(v)}\}$.
Therefore, 
\begin{equation*}
  \hat{\sfJ}[\Pi^{(v)}]=\ket{v}\bra{v}+\ket{0}\bra{0}=G,
\end{equation*}
where $\ket{0}=(1,0,\dots,0)^\top$.

Next, for a general $G\in\mathcal{G}$, consider its spectral decomposition
\[
    G=\sum_{j=1}^{m\leq n} c^{(j)}
    \ket{v^{(j)}}\bra{v^{(j)}} +\ket{0}\bra{0},
    \qquad
    \sum_{j=1}^{m} c^{(j)} \leq 1,
  \]
 Define a set of POVMs through the corresponding eigenvector $\ket{v^{(j)}}\bra{v^{(j)}}$: 
  \[
    \Pi^{(j)}=\{\Pi^{(j)}_k\}_{k},
    \qquad
    \sum_k\Pi^{(j)}_k=I_2 \quad \text{for each j},
  \]
where $\Pi^{(j)}=\{\Pi^{(j)}_k\}_k$ is the PVM associated with $\sum_i v_i^{(j)}\hat L_i$. By construction, we have
  \begin{equation}
    \sum_{j,k} c^{(j)}\Pi_{k}^{(j)}=\sum_j c^{(j)} I_2 =I_2.
  \end{equation}
Hence, we can define a randomized POVM defined by randomization of PVMs as $\Pi=\{c^{(j)}\Pi_k^{(j)}\}$, and the following equality holds for $\hat \sfJ_\mathrm{ex}[\Pi]$:
  \begin{equation*}
    \hat\sfJ_\mathrm{ex}[\Pi]=\sum_j c^{(j)}\hat\sfJ_\mathrm{ex}[\Pi^{(j)}]=\sum_j c^{(j)}\ket{v^{(j)}}\bra{v^{(j)}}+\ket{0}\bra{0}=G.
   \end{equation*}
Therefore, the two sets are equivalent.
\end{proof}

Finally, by applying Lemma~\ref{lemma:GMIneq}, we obtain a Bayesian version of the GM bound, which is the main result of this paper.

\begin{theorem}[Bayesian Gill-Massar bound] \label{thm:BGM}
For any  weight $\sfW$ and the corresponding  $\sfK$, the Bayes risk is lower bounded as 
  \begin{equation}
\begin{split}
      \cBayes
      \geq&
      \cBGM,\\
      \cBGM:=&
      \sfTr{\sfW
      (\muTwo-\ket{\overline{\theta}}\bra{\overline{\theta}})}\\
       &-(d-1)\left\|
        \sqrt{\sfK-\ket{\overline{\theta}}\bra{\overline{\theta}}}\sfW\sqrt{\sfK-\ket{\overline{\theta}}\bra{\overline{\theta}}}
      \right\|_\infty, 
\end{split}
  \end{equation}
where $\|\cdot\|_\infty$ denotes the maximum eigenvalue.
\end{theorem}

\begin{proof}
Since the first term in \eqref{eq:alt_CMI} is constant, we minimize the rest. 
  \begin{align*}
    &\min_{\Pi}-  \sfTr{\sfW_\mathrm{ex} \sfJ_\mathrm{ex}[\Pi]}+1\\
    &=  -\max_{\Pi}\sfTr{ R_\mathrm{ex}\sfW_{\mathrm{ex}}
    R_\mathrm{ex}^\top \hat \sfJ_\mathrm{ex}[\Pi]}+1\\
    &\geq -\max_{G\in\mathcal{G}}\sfTr{ R_\mathrm{ex}\sfW_{\mathrm{ex}}    R_\mathrm{ex}^\top G}+1\\
    &=-\sfTr{\sfW\ket{\overline \theta}\bra{\overline \theta}}\\
    &\quad-\max_{X\in\cL_+(\mathbb{R}^n),\Tr{X}\leq d-1}\Tr{X\sqrt{\sfK-\ket{\overline \theta}\bra{\overline \theta}} \sfW \sqrt{\sfK-\ket{\overline \theta}\bra{\overline \theta}}}\\
    &=-\sfTr{\sfW\ket{\overline{\theta}}\bra{\overline{\theta}}}-(d-1)\left\|   \sqrt{\sfK-\ket{\overline{\theta}}\bra{\overline{\theta}}}\sfW\sqrt{\sfK-\ket{\overline{\theta}}\bra{\overline{\theta}}}
      \right\|_\infty .
  \end{align*}
The first inequality is due to Lemma \ref{lemma:GM2dim}. The equality below follows by rewriting the optimization explicitly. The final equality is proven by solving the SDP problem, which is given by Lemma~\ref{lemma:maximumAX}.
\end{proof}
\begin{lemma}\label{lemma:maximumAX}
Given a positive semidefinite matrix $A\in \mathcal{L}_+(\mathbb{R}^n)$,
  \begin{equation}
    \max_X \{\sfTr{AX}\mid X\in \mathcal{L}_+(\mathbb{R}^n),\ \sfTr{X}\leq d\}
    =d\|A\|_\infty.
  \end{equation}
\end{lemma}
\begin{proof}
Consider the equivalent problem $\min_X -\sfTr{AX}$ subject to $X\geq 0$ and $\sfTr{X}=d$. Introduce Lagrange multipliers $\Lambda\succeq 0$ and $\nu\in\mathbb{R}$ and define
  \begin{equation}
  \begin{split}
      g(\Lambda,\nu)
      =&\inf_X\bigl(-\sfTr{AX}-\sfTr{X\Lambda}+\nu(\sfTr{X}-d)\bigr)\\
      =&\begin{cases*}
        -\nu d & if $\Lambda=\nu X-A$\\
        -\infty & \text{Otherwise}
      \end{cases*}. 
  \end{split}
  \end{equation}
The dual problem becomes
  \[
    \max_{\nu}\{-\nu d\mid \nu I-A\geq 0\},
  \]
which can be solved straightforward by diagonalizing $A$. The optimum is attained at $\nu=\|A\|_\infty$, and an optimal primal solution $X^\ast$ is supported on the eigenspace corresponding to $\|A\|_\infty$. Reversing the sign yields $\max_X\sfTr{AX}=d\|A\|_\infty$.
\end{proof}

\subsection{Attainability}
We now specialize the B-GM bound to the qubit case. Lemma~\ref{lemma:GM2dim} follows that, the B-GM bound is attainable whenever the optimal POVM saturates the maximal eigenvector. For qubits, this condition is always satisfied by a PVM. As a consequence, the B-GM bound and the optimal Bayesian risk coincide, yielding a closed-form solution.
\begin{corollary}[Attainability in qubit model]\label{cor:QubitCloseForm}
  For any  weight matrix $\sfW$, the optimal Bayes
  risk in the qubit case admits the closed-form identity
  \begin{equation}
    \cBayes=\cBGM .
  \end{equation}
Let $\sfK$ be the B-SLD-type FIM and $\{\hat L_i\}$ the corresponding orthonormalized operators. Then the optimal estimator is the EAP and, 
the optimal PVM is given by the spectral decomposition of the operator $L_v=\sum_{i=1}^{n} v_i \hat L_i$, where $\ket{v}$ is the eigenvector associated with the largest eigenvalue of
  \[
    \sqrt{\sfK-\ket{\overline \theta}\bra{\overline \theta}}\,
    \sfW\,
    \sqrt{\sfK-\ket{\overline \theta}\bra{\overline \theta}} .
  \]
\end{corollary}

An important practical implication of this result is that the optimal measurement is a PVM. In the qubit case, the optimal PVM can be easily implemented by a standard method. This is in sharp contrast with point estimation, where the optimal measurement may require a randomized measurement or genuinely non-projection-valued measure.

\section{Numerical Examples}
We analyze two examples to demonstrate the behavior of three bounds for various dimension $d$. To simplify, we set $\sfW=I_{n}$. 
\subsection{Model One}
The first model is a two-parameter Pauli model and its high-dimensional embedding. We set the prior as the Beta distribution. The parametrized model is given as follows.
\begin{equation}\label{eq:eg:1}
\begin{split}
    &S_\theta=\frac{1}{d}(I_d+\theta_1\sigma_1(d)+\theta_2\sigma_2(d)),\quad \theta_1^2+\theta_2^2\leq1,\\
    &\pi(\theta)=\mathrm{Beta}(10,2)\propto(\theta_1^2+\theta_2^2)^9(1-\theta_1^2-\theta_2^2),
\end{split}
\end{equation}
where $\sigma_1(d),\sigma_2(d)$ are defined by embedding the Pauli matrices $\sigma_{1},\sigma_{2}$ into the upper-left $2\times2$ block of a $d\times d$ matrix, with zeros elsewhere. As an example, they are given for $d=3$ as
\begin{equation}
  \sigma_1(3)=\begin{pmatrix}
    0&1&0\\
    1&0&0\\
    0&0&0
  \end{pmatrix},\quad 
    \sigma_2(3)=\begin{pmatrix}
    0&-i&0\\
    i&0&0\\
    0&0&0
  \end{pmatrix}.
\end{equation}
We numerically compute the discussed bounds for $d=2,...,14$, and plot it in Fig.~\ref{fig:eg1}.

\begin{figure}[ht]
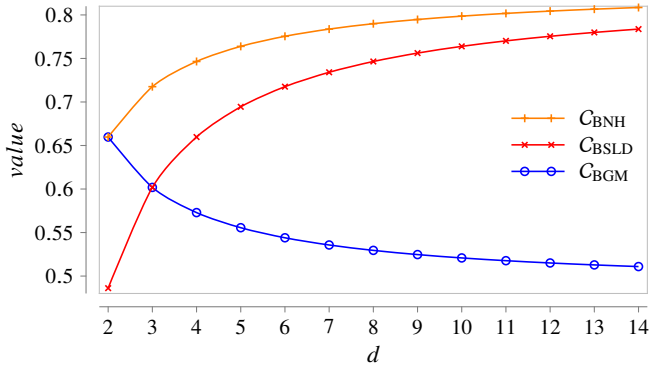

  \tikz\datavisualization[scientific axes=clean,
  every label in legend/.style={node style={font=\normalsize}},
  x axis={label=$d$,length=7.25cm,ticks={step=1},min value=1.8, max value=14.2},
  y axis={label=$value$,length=3.8cm,ticks={style={/pgf/number format/.cd,fixed},step=0.05},max value=0.81,min value=0.48},
  visualize as smooth line/.list={a,b,c},
  legend entry options/default label in legend path/.style={straight label in legend line},
  legend={ at values={x=12.5,y=0.65}},
  a={style={mark=+,mark options=orange,mark size=1.5pt,orange},label in legend={text=$\cBNH$}},
  b={style={mark=x,mark options=red,mark size=1.5pt,red},label in legend={text=$\cBSLD$}},
  c={style={mark=o,mark options=blue,mark size=1.5pt,blue},label in legend={text=$\cBGM$}},
    ]
  data[set=a]{
    x, y
    2 , 0.6597222
    3 , 0.7175926
    4 , 0.7465278
    5 , 0.7638889
    6 , 0.775463
    7 , 0.7837302
    8 , 0.7899306
    9 , 0.7947531
    10 , 0.7986111
    11 , 0.8017677
    12 , 0.8043981
    13 , 0.8066239
    14 , 0.8085317
  }
  data[set=b]{
    x, y
    2 , 0.4861111
    3 , 0.6018519
    4 , 0.6597222
    5 , 0.6944444
    6 , 0.7175926
    7 , 0.734127
    8 , 0.7465278
    9 , 0.7561728
    10 , 0.7638889
    11 , 0.770202
    12 , 0.775463
    13 , 0.7799145
    14 , 0.7837302
  }
  data[set=c]{
    x,y
    2 , 0.6597222
    3 , 0.6018519
    4 , 0.5729167
    5 , 0.5555556
    6 , 0.5439815
    7 , 0.5357143
    8 , 0.5295139
    9 , 0.5246914
    10 , 0.5208333
    11 , 0.5176768
    12 , 0.5150463
    13 , 0.5128205
    14 , 0.5109127
    };
  \caption{Numerical results of $\cBNH$, $\cBGM$ and $\cBSLD$ for the prior distribution in (\ref{eq:eg:1}), and the dimension $d=2,3,\ldots,14.$}
  \label{fig:eg1}
\end{figure}
Figure~\ref{fig:eg1} shows that the B-NH bound and the B-SLD-type bound increase as $d$ increases, whereas the B-GM bound decreases.
In particular, we have $\cBNH=\cBGM$ at $d=2$, and $\cBGM=\cBSLD$ at $d=3$.
The coincidence of the B-SLD-type and B-GM bounds arises from the symmetry of the prior, which can be verified by modifying the prior.
Additionally, the curve of the B-NH bound is always tighter than that of the B-SLD-type bound, which is consistent with Ref.~\cite{bnhbound}.
Finally, $\cBGM \leq \cBSLD$ for $d \geq 4$, indicating a limitation of the B-GM bound in higher-dimensional settings.
\subsection{Model Two}
The second model is parametrized by two hermitian matrices with nonzero entries only on the first off-diagonals. 
\begin{align}\label{eq:eg2}
    &S_\theta=\frac{1}{d}(I_d+\theta_1\omega_1(d)+\theta_2\omega_2(d)),\quad \sqrt{\theta_1^2+\theta_2^2}\leq r(d),\\
    &\pi(\theta)\propto (\theta_1^2+\theta_2^2)^9(r(d)^2-\theta_1^2-\theta_2^2),\ 
r(d):=(2\cos(\frac{\pi}{d+1}))^{-1}.\nonumber
\end{align}
Here, the hermitian matrices $\omega_1(d),\omega_2(d)$ are defined by zero diagonal entries and non-zero elements only on the first off-diagonals, given by 1 and purely imaginary entries, respectively. For example, when $d=3$, they are
\begin{equation}
  \omega_1(3)=\begin{pmatrix}
    0&1&0\\
    1&0&1\\
    0&1&0
  \end{pmatrix},\quad 
    \omega_2(3)=\begin{pmatrix}
    0&-i&0\\
    i&0&-i\\
    0&i&0
  \end{pmatrix}.
\end{equation}
We evaluate the three bounds for $d=3,...,14$ in Fig.~\ref{fig:eg2}.

\begin{figure}[ht]
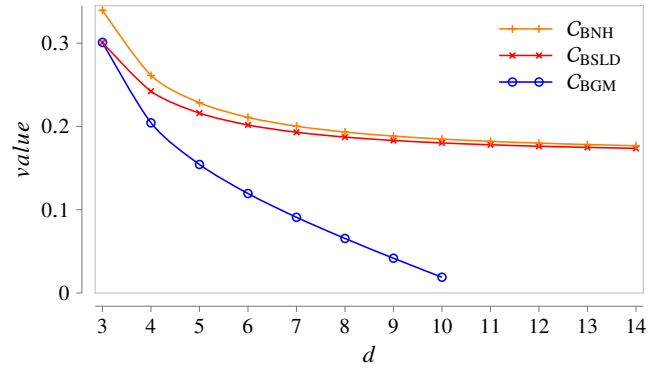

  \tikz\datavisualization[scientific axes=clean,
  every label in legend/.style={node style={font=\normalsize}},
  x axis={label=$d$,length=7.25cm,ticks={step=1},min value=2.85, max value=14.15},
  y axis={label=$value$,length=3.8cm,ticks={style={/pgf/number format/.cd,fixed},step=0.1},max value=0.345,min value=0},
  visualize as smooth line/.list={a,b,c},
  legend entry options/default label in legend path/.style={straight label in legend line},
  legend= north east inside,
  a={style={mark=+,mark options=orange,mark size=1.5pt,orange},label in legend={text=$\cBNH$}},
  b={style={mark=x,mark options=red,mark size=1.5pt,red},label in legend={text=$\cBSLD$}},
  c={style={mark=o,mark options=blue,mark size=1.5pt,blue},label in legend={text=$\cBGM$}},
    ]
  data[set=a]{
    x, y
    3 , 0.3395062
    4 , 0.2613137
    5 , 0.2283952
    6 , 0.2109068
    7 , 0.2003092
    8 , 0.1933145
    9 , 0.188411
    10 , 0.1848156
    11 , 0.1820855
    12 , 0.1799537
    13 , 0.1782508
    14 , 0.1768641
  }
  data[set=b]{
    x, y
    3 , 0.3009259
    4 , 0.2423166
    5 , 0.2160495
    6 , 0.2017584
    7 , 0.1930144
    8 , 0.1872263
    9 , 0.1831715
    10 , 0.1802067
    11 , 0.1779648
    12 , 0.1762231
    13 , 0.1748394
    14 , 0.1737194
  }
  data[set=c]{
    x,y
    3 , 0.3009259
    4 , 0.2043222
    5 , 0.154321
    6 , 0.1194225
    7 , 0.0908876
    8 , 0.0654611
    9 , 0.0417024
    10 , 0.0188966
    };
  \caption{Numerical results of $\cBNH$, $\cBGM$ and $\cBSLD$ for the prior distribution in (\ref{eq:eg2}), and the dimension $d=3,\ldots,14.$ Values of $\cBGM$ for $d \geq 11$ are not shown, as the bound becomes negative and is no longer valid.}
  \label{fig:eg2}
\end{figure}

Different from Model One (\ref{eq:eg:1}), Fig.~\ref{fig:eg2} shows that all bounds decrease as $d$ increases.
As in Model One, $\cBNH=\cBGM$ at $d=2$, $\cBGM=\cBSLD$ at $d=3$, and $\cBGM\leq\cBSLD$ for $d\geq4$.
In contrast, the B-GM bound decreases significantly in this model and becomes negative for $d\geq 11$, indicating a more severe limitation than in Model One. 
In addition, the curves of the B-NH and B-SLD-type bounds become closer as $d$ increases, although a gap always remains.

\section{Discussion and Outlook}
In summary, we introduced a new lower bound, the B-GM bound. We showed that the B-GM bound is attainable for qubit systems and derived the corresponding optimal PVM. Finally, we presented two numerical comparisons of the discussed bounds, which reveal a limitation of the B-GM bound in high-dimensional systems and indicate the equality between the B-NH and B-GM bounds for these models.
As future directions, we plan to investigate the behavior of the B-GM bound in settings with more parameters than those considered in this work. 

\section*{Notes added}
After we have completed this work, we have noticed a recent work by Albarelli \textit{et al.} \cite{albarelli2026measurementincompatibilitybayesianmultiparameter}. 
Ref.~\cite{albarelli2026measurementincompatibilitybayesianmultiparameter} have obtained the closed form of the B-GM bound for a general two parameter qubit model as Eq.~(C13) of \cite{albarelli2026measurementincompatibilitybayesianmultiparameter} together with the optimal estimation strategy. The difference from this work is that they have solved the B-NH bound directly to obtained the attainable bound. In our work, we derived the attainable bound for any qubit model (two- and three-parameter models) by utilizing the B-GM inequality.  

\section*{Acknowledgment}
The work is partly supported by JSPS KAKENHI Grant Number JP24K14816 and ERATO ``Super Quantum Entanglement" (Grant No. JPMJER2402) from JST.



\end{document}